\documentclass[12pt]{article}

\usepackage[scaled=.9]{rsfso}
\usepackage{geometry}
\usepackage{amsmath}
\usepackage{graphicx} 
\usepackage{bickham}
\usepackage{calrsfs}
\usepackage{mathptmx}
\usepackage{helvet}
\usepackage{courier}
\usepackage{amssymb}
\usepackage{amsthm}
\usepackage[T1]{fontenc}
\usepackage[latin9]{inputenc}
\usepackage{makeidx}  
\usepackage{multicol}  
\usepackage[bottom]{footmisc}
\usepackage{relsize}
\usepackage[numbers,sort]{natbib}

\newtheorem{defn}{Definition}
\newtheorem{thm}{Theorem}
\newtheorem{lem}{Lemma}
\usepackage{hyperref}
\hypersetup{colorlinks=true,
  linkcolor=black,
  anchorcolor=black,
  citecolor=black,
  urlcolor=black,
  pdfauthor={Michele Castellana}
}

\providecommand{\nn}{\nonumber}
\providecommand{\be}{\begin{equation}}
  \providecommand{\ee}{\end{equation}}
\providecommand{\bea}{\begin{eqnarray}}
  \providecommand{\eea}{\end{eqnarray}}
\providecommand{\beas}{\begin{eqnarray*}}
  \providecommand{\eeas}{\end{eqnarray*}}

\providecommand{\beni}{\begin{equation*}}
  \providecommand{\eeni}{\end{equation*}}

\providecommand{\bw}{\begin{widetext}}
  \providecommand{\ew}{\end{widetext}}

\makeindex
\date{}

\begin{document}

\title{Rigorous Results for Hierarchical Models of Structural Glasses}
\author{Michele Castellana\footnote{Lewis-Sigler Institute for Integrative Genomics, Princeton University, Princeton, New Jersey 08544, United States.}}
\maketitle

\begin{abstract}
We consider two non-mean-field models of structural glasses built on a hierarchical lattice. First, we consider a hierarchical version of the random energy model (HREM), and we prove the existence of the thermodynamic limit and self-averaging of the free energy. Furthermore, we prove that the infinite-volume entropy is positive in a high-temperature region bounded from below, thus providing an upper bound on the Kauzmann critical temperature. In addition, we show how to improve this bound by leveraging the hierarchical structure of the model. Finally, we introduce  a hierarchical version of the $p$-spin model of a structural glass, and we prove the existence of the thermodynamic limit and self-averaging of the free energy. 
\end{abstract}

\section{Introduction}

Understanding the low-temperature behavior of structural glasses and the nature of their glassy phase is one of the deepest unsolved problems in condensed-matter theory \cite{anderson1995through}: The existence of a Kauzmann transition, a phase transition characterized by the system being frozen in a few low-lying energy states at low temperatures \cite{kauzmann1948nature}, has been the subject of an ongoing debate for a long time now  \cite{biroli2009random}. The development of exactly solvable models mimicking the phenomenology of structural  glasses,  the random energy model (REM)  \cite{derrida1980randomlimit} and the $p$-spin model (PSM) \cite{gross1984simplest}, showed that this transition exists on a mean-field level: For the REM, the Kauzmann transition is characterized by a vanishing entropy \cite{derrida1980randomlimit}, while for the PSM it is characterized by a vanishing complexity, i.e. the logarithm of the number of metastable states \cite{castellani2005spin,biroli2009random,crisanti1995thouless}. Despite the fact that these models reproduce some features of the phenomenology of structural glasses \cite{biroli2009random}, it is still unclear whether the REM and PSM provide a reliable description of the glass transition beyond the mean-field case. After the introduction of these models, further studies suggested how the mean-field physical scenario emerging from the solution of the REM and PSM could be generalized to finite-dimensional systems: A new picture, known as the random first order transition (RFOT) theory \cite{kirkpatrick1989scaling,biroli2009random}, was proposed, suggesting that a Kauzmann transition occurs also for non-mean-field structural glasses. \\

Given that non-mean-field versions of the REM and PSM are hard to solve even with non-rigorous methods, it is natural to study the simplest solvable non-mean-field versions of these models. For ferromagnetic systems, an important role in understanding the non-mean-field scenario has been played by spin systems built on hierarchical lattices \cite{dyson1969existence}: In these models, the renormalization-group (RG) equations emerge in a simple way, and several properties of the ferromagnetic transition can be obtained rigorously  \cite{bleher1973investigation}. To study non-mean-field  structural glasses, it is thus natural to consider hierarchical versions of the REM and PSM \cite{castellana2010hierarchical}: The hierarchical structure of the interactions would then allow for an implementation of RG methods  suitable for studying these systems in the thermodynamic limit. \\

In this paper, we will study a REM and a PSM built on a hierarchical lattice: The hierarchical random energy model (HREM), and the hierarchical $p$-spin model (HPS).  The HREM has been recently proposed  as a non-mean-field model for a structural glass and studied with  perturbative and  numerical methods: These studies suggested that the HREM has a finite-temperature Kauzmann transition characterized by a vanishing entropy at low temperatures \cite{castellana2010hierarchical}, in agreement with the predictions of the RFOT theory of glasses. The HPS, first introduced in this paper, is a candidate  model to study whether the existence of a Kauzmann transition in the PSM \cite{castellani2005spin,biroli2009random,crisanti1995thouless} holds beyond the mean-field scenario \cite{castellana2010hierarchical}.  \\

The paper is organized as follows: In Section \ref{sec1} we define the HREM, we prove the existence of the thermodynamic limit and self-averaging of the free energy, and we derive a mean-field upper bound for the Kauzmann critical temperature. Then, we show how this upper bound can be improved by exploiting the hierarchical structure of the model. In Section \ref{sec2} we introduce the HPS, and we prove the existence of the thermodynamic limit and self-averaging of the free energy. Finally, Section \ref{sec3} is devoted to the discussion of the results and to an outlook on topics of future studies.

\section{Hierarchical Random Energy Model}\label{sec1}

The HREM \cite{castellana2010hierarchical} is a system of $2^{k+1}$ Ising spins $S_i = \pm 1$ labeled by index $i=1,2,\cdots, 2^{k+1}$, whose Hamiltonian is given by the following

\begin{defn}\label{def1}
The Hamiltonian of the hierarchical random energy model (HREM) is defined recursively by the equation
\be
H_{k+1}[\vec{S}] = H_k^1[\vec{S}_1] + H_k^2[\vec{S}_2] +2^{(k+1)\frac{1-\sigma}{2}} \epsilon_{k+1}[\vec{S}],
\ee
where $\vec{S} \equiv \{S_i\}_{1 \leq i \leq 2^{k+1}}$, and $\vec{S}_1 \equiv \{S_i\}_{1 \leq i \leq 2^k}$,  $\vec{S}_2 \equiv \{S_i\}_{2^k+1 \leq i \leq 2^{k+1}}$ are the spins in the left and right half respectively, $H_k^1[\vec{S}_1]$ and $H_k^2[\vec{S}_2]$ are independent,
$\{ \epsilon_{k}[\vec{S}]\}_{k,\vec{S}}$ are IID Gaussian random variables with zero mean and unit variance, and $\sigma$ is a real number. We assign a single-spin  energy $H_0[S] = \epsilon_0[S]$ to each spin, where $\epsilon_0[S]$ is a Gaussian random variable with zero mean and unit variance, and the $2^{k+1}$ single-spin energies are all independent. 
\end{defn}
In Definition \ref{def1} the number $\sigma$ determines how fast spin-spin interactions decrease with distance: The larger $\sigma$, the faster the interaction decrease \cite{dyson1969existence}. In particular, for $\sigma >0$ the HREM is a non-mean-field model, because the variance of the interaction energy between spin blocks $\vec{S}_1$, $\vec{S}_2$, i.e.
\be
\mathbb{E}\big[(2^{(k+1)\frac{1-\sigma}{2}} \epsilon_{k+1}[\vec{S}])^2\big] = 2^{(k+1)(1-\sigma)},
\ee
is subextensive in the system volume $2^{k+1}$ \cite{castellana2010hierarchical}.\\

Unlike the REM, the energies of the HREM are correlated random variables: Given two spin configurations $\vec{S}$, $\vec{S}'$, the energies $H_{k+1}[\vec{S}]$, $H_{k+1}[\vec{S}']$ are not independent. In recent years, Contucci \textit{et al.} showed \cite{contucci2003thermodynamical} that the thermodynamic limit of the quenched free energy of a REM with correlated energies exists under some sufficient conditions on the energy correlations. Namely, given  a family $\{E_{N}[\vec{S}]\}_{\vec{S} \in \Sigma_N}$ of $2^N$ Gaussian random variables with zero mean and unit variance with Hamiltonian $H_N[\vec{S}] \equiv - \sqrt{N} E_N[\vec{S}] $, and given a decomposition $N = N_1 + N_2$ and the projections $\pi_i(\vec{S}) $ of $\vec{S} \in \Sigma_N$ into $\Sigma_{N_i}$ ($i=1,2$), the thermodynamic limit of the quenched free energy exists if 
\be\label{eq26}
\mathbb{E}[H_N[\vec{S}]H_N[\vec{S}']] \leq \mathbb{E}[H_{N_1}[\pi_1(\vec{S})]H_{N_1}[\pi_1(\vec{S}')]] + \mathbb{E}[H_{N_2}[\pi_2(\vec{S})]H_{N_2}[\pi_2(\vec{S}')]]
\ee
for any $\vec{S}, \vec{S}'$ and for any decomposition \cite{contucci2003thermodynamical},  where $\mathbb{E}[]$ denotes the expectation with respect to all random variables. Unfortunately, the condition (\ref{eq26}) does not apply to the HREM: Indeed, for $N = 2^{k+1}$, $N_1 = N_2 = 2^k$ and $\vec{S} = \vec{S}'$, from Definition \ref{def1} we have $\pi_i(\vec{S}) = \vec{S}_i$ ($ i=1,2$) and
\beas
\mathbb{E}[H_{k+1}[\vec{S}]^2] & = & \mathbb{E}[H_k^1[\vec{S}_1]^2] + \mathbb{E}[H_k^2[\vec{S}_2]^2] + 2^{(k+1)(1-\sigma)} \\
 & > & \mathbb{E}[H_k^1[\vec{S}_1]^2] + \mathbb{E}[H_k^2[\vec{S}_2]^2]. 
\eeas

In what follows, we will prove the existence of the thermodynamic limit for the free energy of the HREM with a recursive method that leverages the hierarchical structure of the model.
Let us introduce the partition function
\be\label{eq59}
Z_{k+1} \equiv \sum_{\vec{S}} \exp(- \beta H_{k+1}[\vec{S}]),
\ee
and the free energy
\be\label{eq46}
f_{k+1} \equiv \frac{1}{2^{k+1}} \mathbb{E} \left[\log Z_{k+1} \right],
\ee
where in what follows $\langle \rangle$ denotes the average associated with the \textit{Boltzmannfaktor} (\ref{eq59}), the inverse temperature $\beta$ is a non-negative number, and $\mathbb{E}[]$ denotes the expectation with respect to all random variables.

\subsection{Thermodynamic limit and self-averaging of the free energy}

We will first prove the existence of the thermodynamic limit for  the quenched free energy with the following

\begin{thm}\label{thm5}
If $\sigma > 0$,  the infinite-volume free energy  
\be
f \equiv \lim_{k \rightarrow \infty} f_{k+1}
\ee exists.  
\end{thm}
\begin{proof}
We will prove the existence of the thermodynamic limit of the free energy by using an interpolation method originally introduced for spin glasses \cite{guerra2002thermodynamic,guerra2003broken}: Given a number $0 \leq t \leq 1$, we introduce the interpolating Hamiltonian 
\be \label{eq1}
H_{k+1,t}[\vec{S}] \equiv H_k^1[\vec{S}_1] + H_k^2[\vec{S}_2] + \sqrt{t} \; 2^{(k+1)\frac{1-\sigma}{2}} \epsilon_{k+1}[\vec{S}],
\ee
and the associated partition function and free energy
\bea \label{eq2}
Z_{k+1,t} &\equiv &\sum_{\vec{S}} \exp(-\beta H_{k+1,t}[\vec{S}]), \\ \label{eq3}
\phi_{k+1,t} &\equiv &\frac{1}{2^{k+1}} \mathbb{E} [\log Z_{k+1,t}].
\eea
From Eqs. (\ref{eq46}), (\ref{eq1}), (\ref{eq2}), (\ref{eq3}) we obtain the values of $\phi_{k+1,t}$ for $t = 0,1$
\be\label{eq6}
\phi_{k+1,1} =  f_{k+1},
\ee
and
\bea \label{eq7} 
\phi_{k+1,0} & = &  \frac{1}{2^{k+1}} \mathbb{E} \left[\log \sum_{\vec{S}} \exp[-\beta(H_k^1[\vec{S}_1] + H_k^2[\vec{S}_2])] \right] \\ \nn
 & = &  \frac{1}{2^{k+1}} \Bigg\{ \mathbb{E} \Bigg[\log \sum_{\vec{S}_1} \exp(-\beta H_k^1[\vec{S}_1])\Bigg] +  \mathbb{E} \Bigg[\log \sum_{\vec{S}_2} \exp(-\beta H_k^2[\vec{S}_2]) \Bigg] \Bigg\} \\  \nn
& = & f_k.
\eea
To complete the interpolation, we compute the derivative of $\phi_{k+1,t}$ with respect to $t$.
Given an integer  $n\geq1$ and a function $g: \Sigma_{2^{k+1}}^{n} \rightarrow \mathbb{R}$, we set 
\be\label{eq25}
\langle g \rangle_t \equiv \frac{1}{Z_{k+1,t}^n} \sum_{\{\vec{S} \}}\exp\left(-\beta \sum_{a=1}^n H_{k+1,t}[\vec{S}^a]\right) g(\{ \vec{S} \}),
\ee
where $\{ \vec{S} \} \equiv \{ \vec{S}^1, \cdots, \vec{S}^n \}$. From Eqs. (\ref{eq1}), (\ref{eq2}), (\ref{eq3}) we have
\bea\label{eq4}
\frac{d \phi_{k+1,t}}{dt} &= &    - \frac{\beta }{2 \; 2^{(k+1)\frac{1+\sigma}{2}} \sqrt{t}} \mathbb{E}[\langle \epsilon_{k+1}[\vec{S}]\rangle_t] \\ \nn
&= &    \frac{ \beta^2 }{2\;  2^{(k+1) \sigma }} \hspace{-0.08cm}\Bigg(\hspace{-0.1cm} 1 - \mathbb{E}\Bigg[\frac{1}{Z_{k+1,t}^2} \sum_{\vec{S}^1, \vec{S}^2} \hspace{-0.1cm} \exp[-\beta (H_{k+1,t}[\vec{S}^1] + H_{k+1,t}[\vec{S}^2])] \mathbb{I}(\vec{S}^1 = \vec{S}^2) \Bigg] \Bigg) \\ \nn
&=&    \frac{ \beta^2 }{2\;  2^{(k+1) \sigma }} ( 1 - \mathbb{E}[ \langle \mathbb{I} (\vec{S}^1 = \vec{S}^2) \rangle_t ] ),
\eea
where in the third line of Eq. (\ref{eq4}) we integrated by parts with respect to $\epsilon_{k+1}[\vec{S}]$, and we introduced the function $\mathbb{I}(\vec{S}^1 = \vec{S}^2)$, which is equal to one if $S^1_i = S^2_i \;\; \forall i = 1, \cdots,  2^{k+1}$ and zero otherwise. From Eq. (\ref{eq4}) we obtain
\be \label{eq5}
 \frac{d \phi_{k+1,t}}{dt} \geq 0. 
\ee
We now use Eq. (\ref{eq5}) to compare $f_{k+1}$ with $f_k$. Putting the identity
\be\label{eq8}
\phi_{k+1,1} = \phi_{k+1,0} + \int_0^1 \frac{d \phi_{k+1,t} } {dt} dt 
\ee
together with Eqs.  (\ref{eq6}), (\ref{eq7}), (\ref{eq5}), we obtain
\be\label{eq11}
f_{k+1} \geq f_k.
\ee
To complete the proof, we show that the free energy $f_{k+1}$ is bounded above. To do so, we denote by $\mathbb{E}_{\epsilon_{k+1}}[]$ the expectation value over the random variables $\{ \epsilon_{k+1}[\vec{S}]\}_{\vec{S}}$ in Definition \ref{def1}, and we have
\bea \label{eq10}
f_{k+1} & = & \frac{1}{2^{k+1}} \mathbb{E}[ \mathbb{E}_{\epsilon_{k+1}}[ \log Z_{k+1}]] \\ \nn
& \leq & \frac{1}{2^{k+1}} \mathbb{E}[ \log \mathbb{E}_{\epsilon_{k+1}}[  Z_{k+1}]] \\ \nn
 & = & f_k +  \frac{\beta^2}{2} \, 2^{-(k+1)\sigma},
\eea
where in the second line of Eq. (\ref{eq10}) we used Jensen's inequality, and in the third line we computed explicitly the Gaussian integral over the random energy $\epsilon_{k+1}[\vec{S}]$. We now iterate recursively  Eq. (\ref{eq10}) for $k+1, k, k-1, \cdots, 0$, and we obtain
\bea\label{eq200}
f_{k+1} & \leq & \frac{\beta^2}{2} \, \frac{1}{2^{\sigma}-1} +  \mathbb{E}\Bigg[ \log \sum_{S = \pm1} \exp(-\beta \epsilon_0[S]) \Bigg]  \\ \nn
 & < & \infty,
\eea
where in Eq. (\ref{eq200}) we used the condition $\sigma>0$. Equations (\ref{eq11}), (\ref{eq200}) show that the sequence $k \rightarrow f_{k+1}$ is non-decreasing and bounded above respectively: Thus, $\lim_{k\rightarrow \infty} f_{k+1}$ exists. \end{proof}

We will now prove that the free energy of the HREM is self-averaging in the thermodynamic limit. For the REM, a standard method to prove the self-averaging property of sample-dependent thermodynamic quantities consists in computing the ratio between the variance and the square of the mean, and showing that this ratio vanishes in the thermodynamic limit \cite{mezard2009information}. Unfortunately, this method cannot be applied directly to the HREM because of the presence of correlations between the energy levels \cite{derrida1980randomlimit}. Still, for the HREM sample-to-sample fluctuations can be estimated with an alternative method: By introducing an auxiliary Hamiltonian that interpolates between two HREMs with different disorder realizations, one can control the sample-to-sample fluctuations of the free energy  \cite{albeverio2003lectures,toninelli2002rigorous}, and prove the self-averaging property with the following 

\begin{thm}\label{thm6}
If $\sigma > 0$, then
\be\label{eq106}
\lim_{k\rightarrow \infty} \frac{1}{2^{k+1}} \log Z_{k+1} = f \; \;  \textrm{ with probability } 1. 
\ee
\end{thm}
\begin{proof}
By using the interpolation method mentioned above, one can prove \cite{albeverio2003lectures,toninelli2002rigorous} that
\be\label{eq107}
\mathbb{P}\left(\left| \frac{1}{2^{k+1}} \log Z_{k+1} - f_{k+1} \right| \geq \frac{1}{2^{(k+1)/4}} \right) \leq 2 \exp\left[ - 2^{(k+1)/2} \frac{2^{\sigma}-1}{2 \, 2^{\sigma} \beta^2 }\right]. 
\ee
Equation (\ref{eq107}) and Borel-Cantelli Lemma imply that 
\be
\left| \frac{1}{2^{k+1}} \log Z_{k+1} - f_{k+1} \right| \geq \frac{1}{2^{(k+1)/4}}
\ee
only for finitely many $k$, which proves Eq. (\ref{eq106}). 
\end{proof}

We will now focus on a thermodynamic quantity that plays an important role in models for structural glasses: The infinite-volume entropy. Indeed, a long-standing question in non-mean-field structural glasses is whether there is a freezing transition characterized by a vanishing entropy at low temperatures, and the critical temperature of this transition is known as the Kauzmann transition temperature \cite{kauzmann1948nature,biroli2009random}. \\

To introduce an infinite-volume entropy, we recall that the entropy for a HREM with a finite number of spins is 
\bea\label{eq71}
s_{k+1} &\equiv& - \beta \frac{d f_{k+1}}{d \beta} + f_{k+1}\\ \nn
& = &  \beta \frac{1}{2^{k+1}} \mathbb{E}[\langle H_{k+1}[\vec{S}]\rangle]  + f_{k+1}. 
\eea
The existence of the thermodynamic limit for the free-energy term $f_{k+1}$ in the first line of Eq. (\ref{eq71}) is proven by Theorem \ref{thm5}, while the existence of the $k\rightarrow \infty$ limit of the energy term $d f_{k+1}/d \beta$ does not follow from Theorem \ref{thm5}, and it requires further analysis. Given that $f$ is a convex function of $\beta$, its right (left) derivative exists, thus one possible way to define the infinite-volume entropy would be 
\be\label{eq202}
s_{\pm} \equiv - \beta \left. \frac{d f}{d \beta}\right|_{\pm} + f, 
\ee
where $\left. d f/d \beta\right|_{\pm}$ denotes the right (left) derivative of $f$. It is easy to show that the definition (\ref{eq202}) is not suitable for the method of proof that we will be using in the rest of the paper. For example, suppose that we want to prove a bound for an infinite-volume thermodynamic quantity: To do so, we will first prove the bound for any finite $k$, and then take the $k\rightarrow \infty$ limit, see for example Lemma \ref{lem3}. Since  in general $\left. d f/d \beta\right|_{\pm}$ cannot be written in a simple way as the the infinite-volume limit of $d f_{k+1}/d \beta$, Eq. (\ref{eq202}) does not allow one to write the infinite-volume entropy $s_{\pm}$ as the $k\rightarrow \infty$  limit of finite-volume thermodynamic quantities: Thus,  the definition (\ref{eq202}) is not suitable for our method of proof. A more natural definition of the infinite-volume entropy is the following: Given that $f_{k+1}$ is convex and differentiable and that $f_{k+1}$ converges pointwise to $f$ in an interval, we have  
\be\label{eq201}
\lim_{k\rightarrow \infty} \frac{d f_{k+1}}{d \beta} = \frac{d f}{d \beta},  
\ee
 for every value of $\beta$ where $f$ is differentiable \cite{talagrand2011mean1}. The ensemble of points $\beta$ such that $f$ is differentiable is everywhere except a countable set of exceptional points \cite{talagrand2011mean1}. In general, showing that such set of exceptional points does not exist is not an easy task: For example, for the Sherrington-Kirkpatrick model of a spin glass \cite{kirkpatrick1975solvable} the proof that the infinite-volume free energy is differentiable everywhere requires a detailed knowledge of the exact solution of the problem \cite{panchenko2008differentiability}. Assuming that $\beta$ is not one of the above exceptional points, in what follows we will define the infinite-volume entropy for a given $\beta$ as 
\be\label{eq98}
s \equiv \lim_{k \rightarrow \infty} s_{k+1}.
\ee 
According to this definition, $s$ is simply given by the $k \rightarrow \infty$ limit of finite-volume quantities, see Eqs. (\ref{eq71}), (\ref{eq98}): To prove---for example---a bound for $s$, we can simply prove the bound for $s_{k+1}$ first, and then take the $k\rightarrow \infty$ limit.

\subsection{Upper bound on the Kauzmann temperature}

We now establish two bounds for the Kauzmann temperature. For $\sigma > 0$, we set
\be\label{eq90}
\varphi(\beta) \equiv \mathbb{E}\Bigg[ \log \sum_{S = \pm1} \exp\Bigg[ - \beta \Bigg(\frac{2^\sigma}{2^\sigma-1} \Bigg)^{1/2} \epsilon_0[S]\Bigg] \Bigg],
\ee
and we have the following
\begin{thm} \label{thm3} (Mean-field bound for the Kauzmann temperature) For $\sigma > 0$ and $\beta$ such that the infinite-volume entropy $s$ exists, we have
\be\label{eq80}
s \geq \log 2 - \beta  \bigg(\frac{2^\sigma}{2^\sigma-1} 2 \log 2\bigg)^{1/2}. 
\ee
Hence, if there exists an inverse Kauzmann temperature $\beta_K$ such that $s = 0 $ for $\beta = \beta_K$, then 
\bea\label{eq81}
\beta_K \geq  \bigg(\frac{2^\sigma-1 }{2^\sigma}  \frac{\log 2}{2} \bigg)^{1/2}.
\eea
\end{thm}

\begin{thm} \label{thm4} (Improvement over the mean-field bound for the Kauzmann temperature) For $\sigma > 0$ and $\beta$ such that the infinite-volume entropy $s$ exists, we have
\be\label{eq69}
s \geq \varphi(\beta) - \beta  \bigg(\frac{2^\sigma}{2^\sigma-1} 2 \log 2\bigg)^{1/2}. 
\ee
Hence, if there exists an inverse Kauzmann temperature $\beta_K$ such that $s = 0 $ for $\beta = \beta_K$, then 
\be\label{eq74}
\beta_K \geq \beta_{\mathlarger{\ast}}, 
\ee
where $\beta_{\mathlarger{\ast}}$ is the unique solution of
\be\label{eq57}
\varphi(\beta) - \beta \bigg(\frac{2^\sigma}{2^\sigma-1} 2 \log 2\bigg)^{1/2} = 0.
\ee
\end{thm}

For the sake of clarity, we note that the lower bounds (\ref{eq81}), (\ref{eq74}) for the inverse Kauzmann critical temperature $\beta_K$ are upper bounds for the Kauzmann critical temperature $T_K \equiv 1/\beta_K$. \\

Before proving Theorems \ref{thm3}, \ref{thm4}, it is important to point out that we refer to Eq. (\ref{eq81}) as the mean-field bound because the inverse critical temperature in the right-hand side (RHS) of Eq. (\ref{eq81}) is proportional to the inverse critical temperature in the mean-field approximation. Indeed, the mean-field approximation of the HREM can be easily obtained by assuming that the energy levels $\{H_{k+1}[\vec{S}]\}_{\vec{S}}$ are independent: In this case, the HREM reduces to a REM with a rescaled inverse critical temperature \cite{castellana2010hierarchical}
\be
\beta_c =  \bigg( \frac{2^{\sigma}-1}{2^\sigma} 2 \log 2 \bigg)^{1/2},  
\ee
which is proportional to the RHS of Eq. (\ref{eq81}) up to a constant factor independent of $\sigma$. Hence, Theorem \ref{thm3}  shows that---up to a constant factor---the mean-field critical temperature is an upper bound for the critical temperature of the system. In addition, in what follows we will show that bound (\ref{eq74}) provides an improvement over bound (\ref{eq81}). \\

Let us now prove Theorems \ref{thm3}, \ref{thm4}. The proofs are based on a lower bound for the infinite-volume entropy, and they will be split into separate Lemmas. First, we prove the following lower bound for the entropy

\begin{lem}\label{lem3}
Given $\sigma >0$ and $\beta$ such that the infinite-volume entropy $s$ exists, then
\be\label{eq66}
s \geq f- \beta \bigg(\frac{2^\sigma}{2^\sigma-1} 2 \log 2\bigg)^{1/2} .  
\ee
\end{lem}

\begin{proof}
We have
\be\label{eq100}
\frac{1}{2^{k+1}}\mathbb{E}[\langle H_{k+1}[\vec{S}]\rangle] \geq \frac{1}{2^{k+1}}\mathbb{E}\Big[\min_{\vec{S}} H_{k+1}[\vec{S}]\Big]. 
\ee
We now use a standard inequality for Gaussian random variables \cite{talagrand2011mean2}. Consider $M$ Gaussian random variables $\{ g_i \}_{1 \leq i \leq M}$ with $M >1$, $\mathbb{E}[g_i] = 0$ and $\mathbb{E}[g_i^2] = \tau^2 > 0$: We do not assume that $\{g_i\}$ are independent. We have 
\be\label{eq99}
\mathbb{E}\Big[\min_i g_i\Big] > - \tau (2 \log M)^{1/2}.
\ee
We now use Eq. (\ref{eq99}) with $\{ g_i \} = \{ H_{k+1}[\vec{S}] \}$, $M = 2^{2^{k+1}}$ and 

\bea\label{eq60}
\tau^2 & = & \mathbb{E}[H_{k+1}[\vec{S}]^2] \\ \nn
& = & \sum_{l=0}^{k+1} 2^{k+1-l} 2^{l(1-\sigma)} \\ \nn
& \leq & 2^{k+1} \frac{2^\sigma}{2^\sigma-1},
\eea 
and we obtain
\be\label{eq101}
\frac{1}{2^{k+1}}\mathbb{E}\Big[\min_{\vec{S}} H_{k+1}[\vec{S}]\Big] \geq- \bigg(\frac{2^\sigma}{2^\sigma-1} 2 \log 2\bigg)^{1/2}.  
\ee
 Equations (\ref{eq71}), (\ref{eq100}), (\ref{eq101}) show that
\be\label{eq102}
s_{k+1} \geq - \beta \bigg(\frac{2^\sigma}{2^\sigma-1} 2 \log 2\bigg)^{1/2} + f_{k+1}. 
\ee
Finally, we  take the $k \rightarrow \infty$ limit of both sides of Eq. (\ref{eq102}), we use the hypothesis that $\lim_{k \rightarrow \infty}s_{k+1}$ exists and Theorem \ref{thm5}, and we obtain Eq. (\ref{eq66}). 
\end{proof}

We now prove a lower bound for the infinite-volume free energy in Eq. (\ref{eq66}) by using a strategy recently proposed in \cite{castellana2014free} for hierarchical models of spin glasses. With this method, the energy $\epsilon_{k+1}[\vec{S}]$ is reabsorbed into two random energies $e_1[\vec{S}_1]$, $e_2[\vec{S}_2]$ for the left and right-half of the spins respectively: As we will show in the following, this method improves over the simple inequality $f \geq \log 2$

\begin{lem}\label{lem4}
For $\sigma > 0$, the infinite-volume free energy satisfies  
\bea\label{eq48}
f  \geq  \varphi(\beta),
\eea
where $\varphi(\beta)$ is given by Eq. (\ref{eq90}). 
\end{lem}
\begin{proof}
Given a number $x$, we set
\bea \label{eq40}
H_{k+1,t}[\vec{S}] & \equiv &  \sqrt{t} x \, 2^{(k+1)\frac{1-\sigma}{2}} \epsilon_{k+1}[\vec{S}] + \sqrt{1-t} \frac{x}{\sqrt{2}}2^{(k+1)\frac{1-\sigma}{2}}  (e_1[\vec{S}_1] + e_2[\vec{S}_2] )  +\\\nn
&  & + H_k^1[\vec{S}_1] + H_k^2[\vec{S}_2], \\ \label{eq41}
Z_{k+1,t} & \equiv & \sum_{\vec{S}} \exp(-\beta H_{k+1,t}[\vec{S}] ),\\ \label{eq42}
\phi_{k+1,t}(x)  & \equiv & \frac{1}{2^{k+1}} \mathbb{E}[\log Z_{k+1,t}], 
\eea
where $\{ e_1[\vec{S}_1] \}_{\vec{S}_1}$, $\{ e_2[\vec{S}_2] \}_{\vec{S}_2}$ are IID Gaussian random variables with zero mean and unit variance which are independent of all other random variables. Let us now proceed with the interpolation: From Eqs. (\ref{eq46}), (\ref{eq40}), (\ref{eq41}), (\ref{eq42}) we have 
\be
\phi_{k+1,1}(1) = f_{k+1},
\ee
and 
\bea \label{eq45}\nn
\phi_{k+1,0}(x)&  = & \frac{1}{2^k} \mathbb{E}\Bigg[\log  \sum_{\vec{S}} \exp\Bigg\{ \hspace{-0.15cm}- \beta \Bigg[H_{k-1}^1[\vec{S}_1] + H_{k-1}^2[\vec{S}_2]  +   \left(1+\frac{x^2}{2^\sigma}\right)^{\frac{1}{2}} \hspace{-0.1cm} 2^{k \frac{1-\sigma}{2}} \epsilon_k[\vec{S}] \Bigg]  \Bigg\} \Bigg]  \\
 & =  &\phi_{k,1} \Big(\Big(1+\frac{x^2}{2^\sigma}\Big)^{\frac{1}{2}} \Big).
\eea
The derivative of $\phi_{k+1,t}(x)$ with respect to $t$ reads
\bea \label{eq43}
\frac{d \phi_{k+1,t}(x)}{dt} = \frac{\beta^2 x^2}{2 \; 2^{(k+1)\sigma}} \mathbb{E}\left[ \left\langle \frac{1}{2} \left[\mathbb{I}(\vec{S}^1_1 = \vec{S}^2_1 ) + \mathbb{I}(\vec{S}^1_2 = \vec{S}^2_2 )\right] - \mathbb{I}(\vec{S}^1 = \vec{S}^2) \right\rangle_t \right] ,
\eea
where $\langle \rangle_t$ is given by Eq. (\ref{eq25}), $\vec{S}^a = \{ S^a_1, \cdots, S^a_{2^{k+1}} \}$ denotes the spin configuration of replicas $a=1,2$, and $\vec{S}^a_1 = \{ S^a_1, \cdots, S^a_{2^k} \}$, $\vec{S}^a_2 = \{ S^a_{2^k+1}, \cdots, S^a_{2^{k+1}} \}$ are the projections onto the left and right half respectively. Equation (\ref{eq43}) and the inequality 
\be
\frac{1}{2} \left[\mathbb{I}(\vec{S}^1_1 = \vec{S}^2_1 ) + \mathbb{I}(\vec{S}^1_2 = \vec{S}^2_2 )\right] \geq \mathbb{I}(\vec{S}^1 = \vec{S}^2)
\ee
imply
\be \label{eq44}
\frac{d \phi_{k+1,t}(x)}{dt}  \geq 0 . 
\ee
Equations  (\ref{eq45}), (\ref{eq44}) give 
\be\label{eq47}
\phi_{k+1,1}(x) \geq \phi_{k,1} \bigg(\bigg(1+\frac{x^2}{2^\sigma}\bigg)^{\frac{1}{2}} \bigg). 
\ee
We now use Eq. (\ref{eq47}) recursively
\bea \label{eq49}
f_{k+1}  & \geq & \phi_{k,1} \bigg(\bigg(1+\frac{1}{2^\sigma}\bigg)^{\frac{1}{2}} \bigg) \\  \nn
& \geq & \phi_{k-1,1} \bigg(\bigg(1+\frac{1}{2^\sigma}+ \frac{1}{2^{2\sigma}}\bigg)^{\frac{1}{2}} \bigg) \\\nn
& \geq & \cdots \\ \nn
& \geq &  \phi_{1,0} \bigg(\bigg(1+\sum_{l=1}^k\frac{1}{2^{l\sigma}}\bigg)^{\frac{1}{2}} \bigg) \\  \nn
& = & \mathbb{E}\Bigg[ \log \sum_{S = \pm1} \exp\Bigg[ - \beta \Bigg(1 + \sum_{l=1}^{k+1} \frac{1}{2^{l \sigma}} \Bigg)^{1/2} \epsilon_0[S]\Bigg] \Bigg].
\eea
Equation (\ref{eq48}) follows by taking the $k \rightarrow \infty$ limit of both sides of Eq. (\ref{eq49}) and by using Theorem \ref{thm5} and Eq. (\ref{eq90}). 
\end{proof}

To obtain a bound on the Kauzmann temperature, we now establish two  properties of the lower bound for the entropy obtained in Lemma \ref{lem4}  

\begin{lem}\label{lem2}
For $\sigma > 0$, 
we have
\be\label{eq86}
\frac{d}{d\beta}\bigg[ \varphi(\beta) - \beta \bigg(\frac{2^\sigma}{2^\sigma-1} 2 \log 2\bigg)^{1/2}  \bigg] \leq 0. 
\ee
 and Eq. (\ref{eq57})  has a unique solution, that we will denote by  $\beta_{\mathlarger{\ast}}$. 
\end{lem}
\begin{proof}
Equation (\ref{eq86}) follows from 
\bea \label{eq58}
\frac{d \varphi(\beta)}{d \beta} &  = &  - \bigg(\frac{2^\sigma}{2^\sigma-1} \bigg)^{1/2} \mathbb{E}[\langle \epsilon_0 [S] \rangle_0 ] \\ \nn
& \leq &  - \bigg(\frac{2^\sigma}{2^\sigma-1} \bigg)^{1/2} \mathbb{E}\Big[\min_{S} \epsilon_0 [S]  \Big]\\\nn
& \leq &   \bigg(\frac{2^\sigma}{2^\sigma-1} 2 \log 2\bigg)^{1/2},
\eea
where $\langle \rangle_0$ denotes the average associated with the \textit{Boltzmannfaktor}  
\be\label{eq50}
\zeta \equiv \sum_{S = \pm1} \exp\left[ - \beta  \bigg(\frac{2^\sigma}{2^\sigma-1} \bigg)^{1/2} \epsilon_0[S]\right],
\ee
and in the last line of  Eq. (\ref{eq58}) we used Eq. (\ref{eq99}).  We now prove that there is a unique solution to Eq. (\ref{eq57}): First, from Eq. (\ref{eq90}) we have 
\be\label{eq55}
\varphi(0) = \log 2.
\ee
Second, we consider $\beta > 0$ and we have
\be\label{eq104}
\exp\left[ - \beta  \bigg(\frac{2^\sigma}{2^\sigma-1} \bigg)^{1/2} \min_{S} \epsilon_0[S]\right] \leq \zeta \leq 2 \exp\left[ - \beta  \bigg(\frac{2^\sigma}{2^\sigma-1} \bigg)^{1/2} \min_{S} \epsilon_0[S]\right]. 
\ee
We now take the logarithm of both sides of Eq. (\ref{eq104}), we divide by $\beta$ and we take the expectation: By using Eqs. (\ref{eq90}), (\ref{eq50}) we obtain 
\be\label{eq91}
- \bigg(\frac{2^\sigma}{2^\sigma-1} \bigg)^{1/2} \mathbb{E}\Big[\min_{S} \epsilon_0[S]\Big] \leq \frac{\varphi(\beta)}{\beta}  \leq - \bigg(\frac{2^\sigma}{2^\sigma-1} \bigg)^{1/2}\mathbb{E}\Big[ \min_{S} \epsilon_0[S]\Big]  + \frac{\log 2}{\beta}. 
\ee
Equation (\ref{eq91}) shows that $\lim_{\beta \rightarrow \infty} (\varphi(\beta)/\beta)$ exists, and that it is given by 
\be\label{eq52}
\lim_{\beta \rightarrow \infty} \frac{\varphi(\beta)}{\beta} = - \bigg(\frac{2^\sigma}{2^\sigma-1} \bigg)^{1/2} \mathbb{E}\Big[\min_{S} \epsilon_0[S] \Big]. 
\ee
We now estimate the right-hand side (RHS) of Eq. (\ref{eq52}) with Eq. (\ref{eq99}): We obtain
\be\label{eq56}
\lim_{\beta \rightarrow \infty} \frac{1}{\beta}\left[\varphi(\beta)  - \beta \bigg(\frac{2^\sigma}{2^\sigma-1} 2 \log 2\bigg)^{1/2} \right] < 0.
\ee 
Equations (\ref{eq55}), (\ref{eq56}) show that there exists at least one solution to Eq. (\ref{eq57}), and Eq. (\ref{eq86}) proves that the solution is unique. 
\end{proof}

We can now prove Theorems \ref{thm3}, \ref{thm4}

\begin{proof}[Proof of Theorem \ref{thm3}]
From Eq. (\ref{eq90}) we have
\bea\label{eq82}
\varphi(\beta)  & = & \log 2 +   \mathbb{E}\Bigg[ \log \Bigg\{ \frac{1}{2}\sum_{S = \pm1} \exp\Bigg[ - \beta \Bigg(\frac{2^\sigma}{2^\sigma-1} \Bigg)^{1/2} \epsilon_0[S]\Bigg] \Bigg\} \Bigg] \\ \nn
& \geq & \log 2 +   \mathbb{E}\Bigg[  \frac{1}{2}\sum_{S = \pm1} \log \exp\Bigg[ - \beta \Bigg(\frac{2^\sigma}{2^\sigma-1} \Bigg)^{1/2} \epsilon_0[S]\Bigg]  \Bigg] \\ \nn
& = & \log 2,
\eea
where in the second line of Eq. (\ref{eq82}) we used the concavity of the logarithm, and in the third line we used $\mathbb{E}[\epsilon_0[S]]=0$. Equation (\ref{eq80}) is then obtained from Lemmas \ref{lem3}, \ref{lem4} and Eq. (\ref{eq82}). Finally, Eq. (\ref{eq81}) follows from Eq. (\ref{eq80}).
\end{proof}

\begin{proof}[Proof of Theorem \ref{thm4}]
Equation (\ref{eq69}) follows from Lemmas \ref{lem3}, \ref{lem4}.  Equation (\ref{eq74}) follows from  Lemma \ref{lem2}.
\end{proof}

It is easy to show that bound (\ref{eq74}) for the Kauzmann temperature improves over the mean-field bound (\ref{eq81}). Indeed, for a given $\sigma > 0$ and $\beta>0$, a strict inequality holds \cite{hardy1952inequalities} in the second line of Eq. (\ref{eq82}) 
\be\label{eq103}
\varphi(\beta)> \log 2,
\ee
 and the RHS of inequality (\ref{eq74}) is strictly larger than the RHS of inequality (\ref{eq81})
\be
\beta_{\mathlarger{\ast}} > \bigg(\frac{2^\sigma-1 }{2^\sigma}  \frac{\log 2}{2} \bigg)^{1/2}.
\ee
For $\sigma \rightarrow 0$, both the RHS of Eqs. (\ref{eq80}) and (\ref{eq74}) tend to zero: This is in agreement with the physical expectation that the critical temperature should diverge in this limit because the Hamiltonian is superextensive. 

\section{Hierarchical $p$-spin Model}\label{sec2}

We now introduce the HPS: Given an integer $p\geq 3$, the HPS is a system of $p^{k+1}$ Ising spins $S_i = \pm 1$ labeled by index $i=1,2,\cdots, p^{k+1}$, whose Hamiltonian is given by the following
\begin{defn}\label{def2}
The Hamiltonian of the hierarchical $p$-spin model (HPS) is defined recursively by the equation
\be \label{eq20}
H_{k+1}[\vec{S}] = \sum_{r=1}^p H_k^r[\vec{S}_r] - \frac{\sqrt{p!}}{p^{(k+1)(p-2(1-\sigma))/2}} \sum_{i_1 > \cdots > i_p=1}^{p^{k+1}} J_{i_1 \cdots i_p} S_{i_1} \cdots S_{i_p},
\ee 
where $\vec{S} \equiv \{S_i\}_{1 \leq i \leq p^{k+1}}$, and $\vec{S}_r \equiv \{S_i\}_{1 +p^k(r-1) \leq i \leq p^k r}$  are the spins in the $r$-th hierarchical block, $\{ H_k^r[\vec{S}_r] \}_{1\leq r \leq p}$ are independent, $H_0[S_i] \equiv h_i S_i$, $\{ J_{i_1 \cdots i_p} \}_{p^{k+1} \geq  i_1 > \cdots > i_p \geq 1}$ are IID Gaussian random variables with zero mean and unit variance,  $\{ h_i \}_{1 \leq i \leq p^{k+1}}$ are IID Gaussian random variables, and $\sigma$ is a real number.  
\end{defn}

The structure of spin interactions in the HPS is depicted in Fig. \ref{fig1}. Like for the HREM,  the number $\sigma$  in Definition \ref{def2} determines how fast spin-spin interactions decrease with distance: The larger $\sigma$, the faster the interaction decrease. In addition, for $\sigma > 1/2$ the HPS is a non-mean-field system: Indeed, setting 
\be
\eta_{k+1}[\vec{S}] \equiv - \frac{\sqrt{p!}}{p^{(k+1)(p-2(1-\sigma))/2}} \sum_{i_1 > \cdots > i_p=1}^{p^{k+1}} J_{i_1 \cdots i_p} S_{i_1} \cdots S_{i_p},
\ee
for large $k$ the variance of the interaction energy between spin blocks $\vec{S}_1, \cdots, \vec{S}_r$ is 
\be
\mathbb{E}\big[\eta_{k+1}[\vec{S}]^2 \big]  = p^{(k+1)2(1-\sigma)},
\ee
which is subextensive in the system volume $p^{k+1}$.\\

In the large-$p$ limit, the mean-field PSM is known to converge to the REM \cite{derrida1980randomlimit}: In this regard, it is easy to show that this property does not hold for the HPS and the HREM. Indeed, for the HREM the covariance of the interaction energy is 
\be\label{eq204}
\mathbb{E}[\epsilon_{k+1}[\vec{S}]\epsilon_{k+1}[\vec{S}']] = 2^{(k+1)(1-\sigma)} \mathbb{I}(\vec{S} = \vec{S}'), 
\ee
while for the HPS
\bea\label{eq203}
\mathbb{E}[\eta_{k+1}[\vec{S}]\eta_{k+1}[\vec{S}']] &=& p^{(k+1)2(1-\sigma)}\frac{1}{p^{(k+1)p}} \sum_{i_1 \neq \cdots \neq i_p=1}^{p^{k+1}} S_{i_1}S'_{i_1} \cdots S_{i_p}S'_{i_p} \\ \nn
& \overset{k \gg 1}{=} & p^{(k+1)2(1-\sigma)} Q^p \\ \nn
& \overset{p \gg 1}{=} &  p^{(k+1)2(1-\sigma)} \mathbb{I}(\vec{S} = \vec{S}'), 
\eea
where in the second line of Eq. (\ref{eq203}) we neglected the diagonal terms $i_1 = i_2 \neq \cdots \neq i_p, i_1 = i_2 = i_3 \neq \cdots \neq i_p, \cdots$ in the sum, which are irrelevant for $k \rightarrow \infty$, and $Q \equiv (1/p^{k+1}) \sum_{i=1}^{p^{k+1}} S_i S'_i$ is the overlap between $\vec{S}$ and $\vec{S}'$. Equation (\ref{eq203}) shows that the covariance of the $k+1$-th block interaction energy of the HPS converges for large $p$ to that of the HREM, Eq. (\ref{eq204}), if $k$ is large.  Still, for small $k$ the diagonal terms in the first line of Eq. (\ref{eq203}) cannot be neglected, and the covariance for the HPS differs from that of the HREM. \\

\begin{figure}
  \centering\includegraphics[scale=0.6]{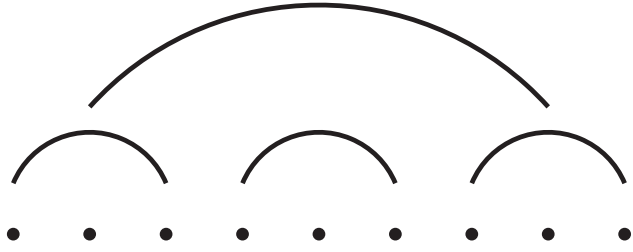}
  \caption{Hierarchical $p$-spin model with $p=3$, $k=1$ and vanishing external magnetic field $h_i$. Each dot represents a spin. The lower arcs represent three-spin interactions between spins below them.  The upper arcs represent three-spin interactions between the spins interacting through the lower arcs. \label{fig1}}
\end{figure}

To prove the existence of the thermodynamic limit, let us introduce the partition function 
\be\label{eq30}
Z_{k+1} \equiv  \sum_{\vec{S}} \exp(-\beta H_{k+1}[\vec{S}]),
\ee
the free energy  
\be\label{eq12}
f_{k+1} =  \frac{1}{p^{k+1}} \mathbb{E} \left[\log  Z_{k+1} \right],
\ee
and prove the following
\begin{thm}\label{thm1}
If $\sigma > 1/2$, the infinite-volume free energy
\be
f \equiv \lim_{k \rightarrow \infty} f_{k+1}
\ee
exists.  
\end{thm}
\begin{proof}
Given $0 \leq t \leq 1$, we introduce the interpolating Hamiltonian and the associated partition function and free energy 

\bea\label{eq15}
H_{k+1,t}[\vec{S}] & \equiv & \sum_{r=1}^p H_k^r[\vec{S}_r] - \sqrt{t} \frac{\sqrt{p!}}{p^{(k+1)(p-2(1-\sigma))/2}} \sum_{i_1 > \cdots > i_p=1}^{p^{k+1}} J_{i_1 \cdots i_p} S_{i_1} \cdots S_{i_p} .
\eea
Following the method used in Theorem \ref{thm5}, it is easy to show that Eq. (\ref{eq15}) implies
\be\label{eq19}
f_{k+1} \geq f_k. 
\ee 
Equation (\ref{eq19}) shows that  the sequence $k \rightarrow f_{k+1}$ is non-decreasing. To show that $\lim_{k\rightarrow \infty} f_{k+1}$ exists, we need to prove that the sequence is also bounded above: This can be done along the same lines as in Theorem \ref{thm5}. Let us denote by $\mathbb{E}_J[]$ the average over the random variables  at the $k+1$-th hierarchical level $\{ J_{i_1 \cdots i_p} \}_{p^{k+1} \geq i_1 > \cdots > i_p \geq 1}$ in Definition \ref{def2}. From Eq. (\ref{eq12}) we have
\bea \label{eq23}
f_{k+1} & = &   \frac{1}{p^{k+1}} \mathbb{E} [ \mathbb{E}_J[\log  Z_{k+1} ]] \\ \nn
& \leq & \frac{1}{p^{k+1}} \mathbb{E} [\log  \mathbb{E}_J[ Z_{k+1} ]] \\  \nn
& = & f_k + \frac{\beta ^2  p^{k+1}(p^{k+1}-1) \cdots (p^{k+1}-(p-1))}{2 \; p^{(k+1)( 2 \sigma-1+p)}}  \\ \nn
& \leq & f_k +  \frac{\beta ^2 }{2 } p^{(k+1)(1-2 \sigma)} , 
\eea
where in the second line of Eq. (\ref{eq23}) we used Jensen's inequality, and in the third line we computed the Gaussian integral over $\{ J_{i_1 \cdots i_p} \}_{p^{k+1} \geq i_1 > \cdots > i_p \geq 1}$. We now iterate Eq. (\ref{eq23}) for $k+1, k, k-1, \cdots, 0$, and we obtain 
\bea  \label{eq24}
f_{k+1}  & \leq & \mathbb{E}[\log( 2\cosh(\beta h))] + \frac{\beta^2 }{2}\frac{1}{p^{2\sigma-1}-1} \\ \nn
& < & \infty,
\eea
where in Eq. (\ref{eq24}) we used the condition $\sigma > 1/2$. Equations (\ref{eq19}), (\ref{eq24}) show that the sequence $k \rightarrow f_{k+1}$ is non-decreasing and bounded above, thus $\lim_{k\rightarrow \infty} f_{k+1} $ exists.
\end{proof}

It is straightforward to prove that the free energy of the HPS is self-averaging
\begin{thm}
If $\sigma>1/2$, then
\be\label{eq105}
\lim_{k\rightarrow \infty} \frac{1}{p^{k+1}} \log Z_{k+1} = f \; \;  \textrm{ with probability } 1. 
\ee
\end{thm}
\begin{proof}
Equation (\ref{eq105}) can be obtained with a step-by-step repetition of the proof of Theorem \ref{thm6}. 
\end{proof}

\section{Conclusions and Outlook}\label{sec3}

In this paper we studied  two non-mean-field models for structural glasses built on a hierarchical lattice.  We first considered a hierarchical version of the random energy model (HREM): The HREM was previously introduced and studied in \cite{castellana2010hierarchical} by means of perturbative and numerical methods, suggesting that the model has a finite-temperature Kauzmann transition, namely a freezing transition characterized by a vanishing entropy at low temperatures \cite{biroli2009random}. For the HREM, we proved the existence of the thermodynamic limit and self-averaging of the free energy. Then, we focused on the possibility that the HREM undergoes a Kauzmann transition. We showed that the infinite-volume entropy is positive in a high-temperature region bounded below by a threshold temperature proportional to the mean-field critical temperature: This implies that if there is a freezing transition then its critical temperature, the Kauzmann temperature \cite{kauzmann1948nature}, must satisfy an upper bound. In addition, by using a method recently proposed in \cite{castellana2014free}, we improved over the above bound by exploiting the hierarchical structure of the model.  Finally, we introduced a $p$-spin model (PSM) built on a hierarchical lattice, the hierarchical $p$-spin model (HPS), a candidate model to study whether the existence of a Kauzmann transition in the mean-field PSM \cite{castellani2005spin,biroli2009random,crisanti1995thouless} holds beyond the mean-field scenario \cite{castellana2010hierarchical}. For the HPS, we proved the existence of the thermodynamic limit and self-averaging of the free energy.\\

As a topic of future research, it would be interesting to study the existence of a finite-temperature Kauzmann transition in the HREM. Indeed, studying the existence of a Kauzmann transition in non-mean-field models of structural glasses is an interesting physical question that has been raising interest  for several decades now \cite{biroli2009random}. Together with the upper bound on the Kauzmann temperature proven in this paper, a lower bound for the transition temperature would provide a rigorous proof of the existence of a Kauzmann transition in a non-mean-field model of a structural glass.

\section*{Acknowledgments}
We would like to thank A. Barra and F. Guerra for useful discussions.  Research supported by NSF Grants PHY--0957573, by the Human Frontiers Science Program, by the Swartz Foundation, and by the W. M. Keck Foundation.

\bibliographystyle{unsrtnat}
\bibliography{bibliography}

\end{document}